\documentclass[copyright,creativecommons]{eptcs}
\providecommand{\event}{CCA 2010} % Name of the event you are submitting to
\usepackage{amsthm,amstext,amsfonts,enumerate}
\usepackage{amssymb}
\usepackage{latexsym}
\usepackage{makeidx}
\usepackage{amsmath}

\newcommand{\N}{{\mathbb N}}

\newcommand{\TN}{2^{\N}}

\newcommand{\B}{{\mathcal B}}
\newcommand{\T}{{\mathcal T}}
\newcommand{\C}{{\mathcal C}}

\newcommand{\ob}{\hat{b}}
\newcommand{\fr}{{^\frown}}
\newcommand{\la}{\langle}
\newcommand{\ra}{\rangle}
\newcommand{\pz}{$\Pi^0_1$\ }
\newcommand{\ml}{Martin-L\"{o}f\ }

\theoremstyle{plain}
\newtheorem{theorem}{Theorem}[section]
\newtheorem{lemma}[theorem]{Lemma}
\newtheorem{claim}[theorem]{Claim}
\newtheorem{corollary}[theorem]{Corollary}
\newtheorem{example}{Example}[section]
\newtheorem{definition}[theorem]{Definition}
\newtheorem{proposition}[theorem]{Proposition}

\newtheorem{remark}[theorem]{Remark}
\begin{document}

\title{Effective Capacity and \\Randomness of Closed Sets}
\author{Douglas Cenzer\thanks{This research was partially supported by NSF
grants DMS 0532644 and 0554841 and 652372.}
\institute{Department of Mathematics, University of Florida,
Gainesville, Florida, USA}
\email{cenzer@ufl.edu}
\and
Paul Brodhead
\institute{Virginia State University,
Petersburg, Virginia, USA}
\email{ brodhe@gmail.com}
}
\def\titlerunning{Effective Capacity}
\def\authorrunning{D. Cenzer \& P. Brodhead}

%Phone: 352-392-0281\ \ Fax: 352-392-8357

\maketitle

\begin{abstract}
We investigate the connection between measure and capacity for the
space $\C$ of nonempty closed subsets of $\TN$.  For any computable measure
$\mu^*$, a computable capacity $\T$ may be defined by letting $\T(Q)$ be the
measure of the family of closed sets $K$ which have nonempty
intersection with $Q$. We prove an effective version of Choquet's
capacity theorem by showing that every computable capacity may be obtained from
a computable measure in this way. We establish conditions that
characterize when the capacity of a random closed set equals zero or
is $>0$.  We construct for certain measures an effectively closed set
with positive capacity and with Lebesgue measure zero.

%\textbf{Keywords:} Computability, Randomness, $\Pi^0_1$ Classes
\end{abstract}

%\renewcommand{\thefootnote}{}
%\footnote{$^{\star}$
%\vspace{-1.3 cm}

\section{Introduction}The study of algorithmic randomness has been an
active area of research in recent years.  The basic problem is to
quantify the randomness of a single real number. Here we think of a
real $r \in [0,1]$ as an infinite sequence of 0's and 1's, i.e as an
element in $\TN$. There are three basic approaches to algorithmic
randomness: the measure-theoretic approach of \ml tests,
the incompressibility approach of Kolmogorov complexity, and the
betting approach in terms of martingales.  All three approaches have
been shown to yield the same notion of (algorithmic) randomness. The
present paper will consider only the measure-theoretic approach.  A
real $x$ is \ml random if for any effective sequence $S_1,
S_2, \dots$ of c. e. open sets with $\mu(S_n) \leq 2^{-n}$, $x \notin
\cap_n S_n$.  For background and history of algorithmic randomness we
refer to \cite{DH:book,Nies:book}.

In a series of recent papers \cite{BCRW08,BCD06,BCDW08,BCR06},
G. Barmpalias, S. Dashti, R. Weber and the authors have defined a
notion of (algorithmic) randomness for closed sets and continuous
functions on $2^{\N}$.  Some definitions are needed.  For a finite
string $\sigma \in\{0,1\}^n$, let $|\sigma| = n$.  For two strings
$\sigma,\tau$, say that $\tau$ \emph{extends} $\sigma$ and write
$\sigma \sqsubseteq \tau$ if $|\sigma| \leq |\tau|$ and $\sigma(i) =
\tau(i)$ for $i < |\sigma|$. For $x \in \TN$, $\sigma \sqsubset x$
means that $\sigma(i) = x(i)$ for $i < |\sigma|$. Let $\sigma^{\frown}
\tau$ denote the concatenation of $\sigma$ and $\tau$ and let
$\sigma^{\frown} i$ denote $\sigma^{\frown}(i)$ for $i=0,1$. Let
$x\lceil n =(x(0),\dots,x(n-1))$. Two reals $x$ and $y$ may be coded
together into $z = x \oplus y$, where $z(2n) = x(n)$ and $z(2n+1) =
y(n)$ for all $n$.
For a finite string $\sigma$, let $I(\sigma)$ denote $\{x \in
\TN:\sigma \sqsubset x\}$. We shall call $I(\sigma)$ the
\emph{interval} determined by $\sigma$. Each such interval is a clopen
set and the clopen sets are just finite unions of intervals. We let
$\B$ denote the Boolean algebra of clopen sets.

Now a closed set $P$ may be identified with a tree
$T_P\subseteq \{0,1\}^*$ where $T_P = \{\sigma: P \cap I(\sigma)
\neq\emptyset\}$. Note that $T_P$ has no dead ends. That is, if
$\sigma\in T_P$, then either $\sigma^{\frown}0 \in T_P$ or
$\sigma^{\frown}1\in T_P$.
% For an arbitrary tree $T \subseteq \{0,1\}^*$, let $[T]$ denote the
set of infinite paths through $T$. It is well-known that $P \subseteq
\TN$ is a closed set if and only if $P = [T]$ for some tree $T$.  $P$
is a $\Pi^0_1$ class, or an effectively closed set, if $P = [T]$ for
some computable tree $T$; equivalently $T_P$ is a \pz set. The
complexity of the closed set $P$ is generally identified with that of
$T_P$.  Thus $P$ is said to be a $\Pi^0_2$ closed set if $T_P$ is
$\Pi^0_2$; in this case $P = [T]$ for some $\Delta^0_2$ tree $T$.  The
complement of a $\Pi^0_1$ class is sometimes called a c.e. open
set. We remark that if $P$ is a $\Pi^0_1$ class, then $T_P$ is a
$\Pi^0_1$ set, but it is not, in general, computable.
There is a natural effective enumeration $P_0, P_1, \dots$ of the
$\Pi^0_1$ classes and thus an enumeration of the c.e. open sets. Thus
we can say that a sequence $S_0,S_1,\dots$ of c.e. open sets is
\emph{effective} if there is a computable function, $f$, such that
$S_n = \TN - P_{f(n)}$ for all $n$. For a detailed development of
$\Pi^0_1$ classes, see~\cite{Ceta}.

It was observed in \cite{BCD06} that there is a natural isomorphism
between the space $\C$ of nonempty closed subsets of $\{0,1\}^{\N}$
and the space $\{0,1,2\}^{\N}$ (with the product topology)
defined as follows. Given a nonempty closed
$Q\subseteq \TN$, let $T = T_Q$ be the tree without dead ends such
that $Q =[T]$. Let $\sigma_0, \sigma_1, \ldots$ enumerate the elements
of $T$ in order, first by length and then lexicographically. We then
define the code $x = x_Q = x_T$ by recursion such that for each $n$,
$x(n) =2$ if both $\sigma_n\fr 0$ and $\sigma_n\fr 1$ are in $T$,
$x(n) =1$ if $\sigma_n\fr 0 \notin T$ and $\sigma_n\fr 1 \in T$, and
$x(n) =0$ if $\sigma_n\fr 0 \in T$ and $\sigma_n\fr 1 \notin T$. For a
finite tree $T \subseteq \{0,1\}^{\leq n}$, the finite code $\rho_T$ is
similarly defined, ending with $\rho_T(k)$ where $\sigma_k$ is the
lexicographically last element of $T \cap \{0,1\}^n$.

We defined in \cite{BCD06} a measure $\mu^*$ on the space ${\mathcal
C}$ of closed subsets of $\TN$ as follows.
\begin{equation}
\mu^*({\mathcal X}) = \mu(\{x_Q:Q \in {\mathcal X}\})
\end{equation}
for any ${\mathcal X} \subseteq {\mathcal C}$ and $\mu$ is the
standard measure on $\{0,1,2\}^{\N}$.  Informally this means that
given $\sigma \in T_Q$, there is probability $\frac13$ that both
$\sigma^{\frown}0 \in T_Q$ and $\sigma^{\frown}1 \in T_Q$ and, for
$i=0,1$, there is probability $\frac13$ that only $\sigma^{\frown}i
\in T_Q$. In particular, this means that $Q \cap I(\sigma) \neq
\emptyset$ implies that for $i=0,1$, $Q \cap I(\sigma^{\frown}i) \neq
\emptyset$ with probability $\frac23$.

Brodhead, Cenzer, and Dashti \cite{BCD06} defined a closed set
$Q\subseteq 2^{\N}$ to be (Martin-L\"{o}f) random if $x_Q$ is
(Martin-L\"{o}f) random.  Note that the equal probability of $\frac13$
for the three cases of branching allows the application of Schnorr's
theorem that \ml randomness is equivalent to prefix-free Kolmogorov
randomness.  Then in \cite{BCD06,BCDW08}, the following results are
proved. No $\Pi^0_1$ class is random but there is a random
$\Delta^0_2$ closed set. Every random closed set contains a random
member but not every member is random. Every random real belongs to
some random closed set. Every random $\Delta^0_2$ closed set contains
a random $\Delta^0_2$ member.  Every random closed set is perfect and
contains no computable elements (in fact, it contains no
$n$-c.e.\ elements).  Every random closed set has measure 0. A random
closed set is a specific type of random recursive construction, as
studied by Graf, Mauldin and Williams \cite{GMW88}. McLinden and
Mauldin \cite{MM09} showed that the Hausdorff dimension of a random
closed set is $log_2(4/3)$.

Just as an effectively closed set in $\TN$ may be viewed as the set of
infinite paths through a computable tree $T \subseteq \{0,1\}^*$, an
algorithmically random closed set in $\TN$ may be viewed as the set of
infinite paths through an algorithmically random tree $T$. Diamondstone and
Kjos-Hanssen \cite{DK09,KH09} give an alternative definition of random closed
sets according to the Galton-Watson distribution and show that this definition
produces the same family of algorithmically random closed sets. The effective
Hausdorff dimension of members of random closed sets is studied in \cite{DK09}.

In the present paper we will examine the notion of computable capacity
and its relation to computable measures on the space $\C$ of nonempty
closed sets.  In section two, we present a family of computable
measures on $\C$ and show how they induce capacities. We define the
notion of computable capacity and present an effective version of
Choquet's theorem that every capacity can be obtained from a measure
$\mu^*$ on the space of closed sets. The main theorem of section three
gives conditions under which the capacity $\T(Q)$ of a $\mu^*$-random
closed set $Q$ is either equal to $0$ or $> 0$.  We also construct a
$\Pi^0_1$ class with Lebesgue measure zero but with positive capacity,
for each capacity of a certain type.

\section{Computable Measure and Capacity on the Space of Closed Sets}

In this section, we describe the hit-or-miss topology on the space
$\C$ of closed sets, we define certain probability measures $\mu_d$ on
the space $\{0,1,2\}^{\N}$ and the corresponding measures $\mu^*_d$ on
the homeomorphic space $\C$. We present an effective version of
Choquet's theorem connecting measure and capacity.

The standard (\emph{hit-or-miss}) topology \cite{D77} (p. 45)
on the space $\C$ of closed sets is given by a sub-basis of sets of two types,
where $U$ is any open set in $2^{\N}$.
\[
V(U) = \{K: K \cap U \neq \emptyset\}; \qquad \qquad W(U) = \{K: K \subseteq U\}
\]

Note that $W(\emptyset) = \{\emptyset\}$ and that $V(\TN) = \C
\setminus \{\emptyset\}$, so that $\emptyset$ is an isolated element
of $\C$ under this topology.  Thus we may omit $\emptyset$ from $\C$
without complications.

A basis for the hit-or-miss topology may be formed by taking finite
intersections of the basic open sets. We want to work with the
following simpler basis.  For each $n$ and each finite tree $A
\subseteq\{0,1\}^{\leq n}$, let

\[
U_A = \{K\in \C: (\forall \sigma \in A) (K \cap I(\sigma) \neq \emptyset)\
\&\ (\forall \sigma \notin A) (K \cap I(\sigma) = \emptyset) \}.
\]
That is,
\[
U_A = \{K \in \C: T_K \cap \{0,1\}^{\leq n} = A\}.
\]
Note that the sets $U_A$ are in fact clopen. That is, for any tree $A
\subseteq \{0,1\}^{\leq n}$, define the tree $A' = \{\sigma \in
\{0,1\}^{\leq n}: (\exists \tau \in \{0,1\}^n \setminus A) \sigma
\sqsubseteq \tau\}$. Then $U_{A'}$ is the complement of $U_A$.

For any finite $n$ and any tree $T \subseteq \{0,1\}^{\leq n}$, define the clopen set
$[T] = \cup_{\sigma \in T} I(\sigma)$. Then $K \cap [T] \neq \emptyset$
if and only if there exists some $A \subseteq \{0,1\}^{\leq n}$ such that $K
\in U_A$ and $A \cap T \neq \emptyset$. That is,
\[
V([T]) = \bigcup\{U_A: A \cap T \neq \emptyset\}.
\]
Similarly, $K \subseteq [T]$
if and only if there exists some $A\subseteq \{0,1\}^n$ such that $K
\in U_A$ and $A \subseteq T$. That is,
\[
W([T]) = \bigcup \{U_A: A \subseteq T\}.
\]
The following lemma can now be easily verified.

\begin{lemma} The family of sets $\{U_A: A \subseteq \{0,1\}^{\leq n}\ \text{a tree}\}$
is a basis of clopen sets for the hit-or-miss topology on $\C$.
\end{lemma}

Recall the mapping from $\C$ to $\{0,1,2\}^{\N}$ taking $Q$ to $x_Q$.
It can be shown that this is in fact a homeomorphism.
(See Axon \cite{Ax10} for details.) Let $\B^*$ be the family of clopen sets
in $\C$; each set is a finite union of basic sets of the form $U_A$ and thus $\B^*$
is a computable atomless Boolean algebra.

\begin{proposition} The space $\C$ of nonempty closed subsets of $\TN$ is homeomorphic
to the space $\{0,1,2\}^{\N}$. Furthermore, the corresponding map from $\B$ to $\B^*$
is a computable isomorphism.
\end{proposition}

Next we consider probability measures $\mu$ on the space $\{0,1,2\}^{\N}$ and the
corresponding measures $\mu^*$ on $\C$ induced by $\mu$.

A probability measure on $\{0,1,2\}^{\N}$ may be defined as in \cite{RSta}
from a function $d: \{0,1,2\}^* \to [0,1]$ such that $d(\lambda) = 1$ and,
for any $\sigma \in\{0,1,2\}^*$,
\[
d(\sigma) = \sum_{i=0}^2 d(\sigma \fr i).
\]
The corresponding measure $\mu_d$ on $\{0,1,2\}^{\N}$ is then defined
by letting $\mu_d(I(\sigma)) = d(\sigma)$. Since the intervals
$I(\sigma)$ form a basis for the standard product topology on
$\{0,1,2\}^{\N}$, this will extend to a measure on all Borel sets.  If
$d$ is computable, then $\mu_d$ is said to be computable. The measure
$\mu_d$ is said to be \emph{nonatomic} or \emph{continuous} if
$\mu_d(\{x\}) = 0$ for all $x \in \{0,1,2\}^{\N}$. We will say that
$\mu_d$ is \emph{bounded} if there exist bounds $b,c \in (0,1)$ such
that, for any $\sigma \in \{0,1,2\}^*$ and $i \in \{0,1,2\}$,
\[
b \cdot d(\sigma) < d(\sigma \fr i) < c \cdot d(\sigma).
\]
It is easy to see that any bounded measure must be continuous. We will
say that the measure $\mu_d$ is \emph{regular} if there exist constants
$b_0,b_1,b_2$ with $b_0+b_1+b_2 = 1$ such that for all $\sigma$ and
for $i \leq 2$, $d(\sigma \fr i) = b_i d(\sigma)$.

Now let $\mu_d^*$ be defined by \[
\mu_d^*({\mathcal X}) = \mu_d(\{x_Q: Q \in {\mathcal X}\}).
\]
Let us say that a measure $\mu^*$ on $\C$ is computable if the
restriction of $\mu^*$ to $\B^*$ is computable.

\begin{proposition}
For any computable $d$, the measure $\mu^*_d$ is a computable measure on $\C$.
\end{proposition}

\begin{proof}
For any tree $A \subseteq \{0,1\}^{\leq n}$, it is easy to see that

\[
K \in U_A \iff \rho_A \sqsubset x_K,
\]
so that $\mu_d^*(U_A) = \mu_d(I(\rho_A))$.
\end{proof}

We are now ready to define capacity.  For details on capacity and random set
variables, see \cite{Ng06}.

\begin{definition}
A \emph{capacity} on $\C$ is a function $\T: \C \to [0,1]$ with
$\T(\emptyset) =0$ such that
\begin{itemize}
\item[(i)] $\T$ is monotone increasing, that is,
\[
Q_1 \subseteq Q_2
  \longrightarrow \T (Q_1) \leq \T(Q_2).
\]
\item[(ii)] $\T$ has the \emph{alternating of infinite order} property, that is,
for $n \geq 2$ and any $Q_1, \dots, Q_n \in \C$
\[
\T(\bigcap_{i=1}^n Q_i) \leq \sum \{(-1)^{|I|+1} \T(\bigcup_{i \in
  I}Q_i): \emptyset \neq I \subseteq \{1,2,\dots,n\} \}.
\]
\item[(iii)] If $Q = \cap_n Q_n$ and $Q_{n+1} \subseteq Q_n$ for all
$n$, then $\T(Q) = lim_{n \to \infty} \T(Q_n)$.
\end{itemize}
\end{definition}

We will also assume, unless otherwise specified, that the capacity
$\T(2^N) = 1$.

We will say that a capacity $\T$ is computable if it is computable on
the family of clopen sets, that is, if there is a computable function $F$ from
the Boolean algebra $\B$ of clopen sets into $[0,1]$ such that
$F(B) = \T(B)$ for any $B \in \B$.
%It follows that the capacity of any
%$\Pi^0_1$ class is upper semi-computable.

Define $\T_{d}(Q)=\mu_d^*(V(Q))$.  That is, $\T_{d}(Q)$ is the
probability that a randomly chosen closed set meets $Q$. Here is
the first result connecting measure and capacity.

\begin{theorem} \label{th1}
If $\mu^{*}_{d}$ is a (computable) probability measure on $\C$, then
$\T_{d}$ is a (computable) capacity.
\end{theorem}

\begin{proof}
Certainly $\T_d(\emptyset) = 0$. The
alternating property follows by basic probability. For (iii), suppose
that $Q = \cap_n Q_n$ is a decreasing intersection. Then by
compactness, $Q \cap K \neq \emptyset$ if and only if $Q_n \cap K \neq
\emptyset$ for all $n$. Furthermore, $V(Q_{n+1}) \subseteq V(Q_n)$ for
all $n$. Thus

\[
\T_d(Q) = \mu^*_d(V(Q)) = \mu^*_d(\cap_n V(Q_n)) = lim_n \mu^*_d(V(Q_n)) =
lim_n\T_d(Q_n).
\]
If $d$ is computable, then $\T_d$ may be computed as
follows.  For any clopen set $I(\sigma_1) \cup \dots \cup
I(\sigma_k)$ where each $\sigma_i \in \{0,1\}^n$, we compute the
probability distribution for all trees of height $n$ and add the
probabilities of those trees which contain one of the
$\sigma_i$.
\end{proof}

Choquet's Capacity Theorem states that any capacity $\T$ is determined by a measure,
that is $\T = \T_d$ for some $d$. See \cite{Ng06} for details. We now give an
effective version of Choquet's theorem.

\begin{theorem} [Effective Choquet Capacity Theorem] \label{th2}
If $\T$ is a computable capacity, then there is a computable measure
$\mu_d^*$ on the space of closed sets such that $\T =
\T_d$. \end{theorem}

\begin{proof}
Given the values $\T(U)$ for all clopen sets $I(\sigma_1)\cup \dots
\cup I(\sigma_k)$ where each $\sigma_i \in \{0,1\}^n$, there is in
fact a unique probability measure $\mu_d$ on these clopen sets such
that $\T = \T_d$ and this can be computed as follows.

Suppose first that $\T(I(i)) = a_i$ for $i < 2$ and note that each
$a_i \leq 1$ and $a_0 + a_1 \geq 1$ by the alternating property.  If
$\T = \T_d$, then we must have $d((0)) + d((2)) = a_0$ and $d((1)) +
d((2)) = a_1$ and also $d((0)) + d((1)) + d((2)) = 1$, so that $d((2))
= a_0 + a_1 - 1$, $d((0)) = 1 - a_1$ and $d((1)) = 1 - a_0$. This will
imply that $\T(I\tau)) = \T_d(I(\tau))$ when $|\tau| = 1$. Now suppose that
we have defined $d(\tau)$ and that $\tau$ is the code for a finite
tree with elements $\sigma_0,\dots,\sigma_n =\sigma$ and thus $d(\tau
\fr i)$ is giving the probability that $\sigma$ will have one or both
immediate successors. We proceed as above. Let $\T(I(\sigma \fr i)) =
a_i \cdot \T(I(\sigma))$ for $i<2$. Then as above $d(\tau \fr 2) =
d(\tau) \cdot (a_0 + a_1 - 1)$ and $d(\tau \fr i) = d(\tau) \cdot (1 -
a_i)$ for each $i$.
\end{proof}

\section{When is $\T(Q)=0$?}

In this section, we compute the capacity of a random closed set under certain
probability measures. We construct a \pz class with measure zero but
with positive capacity.

We say that $K \in \C$ is \emph{$\mu_d^*$-random} if $x_K$ is \ml random
with respect to the measure $\mu_d$. (See \cite{RSta} for details.)

Our next  result shows that the $\T_d$ capacity of a $\mu^*_d$-random closed set
depends on the particular measure.

\begin{theorem}\label{th4}
Let $d$ be the uniform measure with $b_0 = b_1 = b > 0$ and $b_2 =
1-2b > 0$ and let $\ob = 1 - \frac{\sqrt 2}2$.  Then
\begin{itemize}
\item[(a)] If $b \geq \ob$, then
for any $\mu_d^*$-random closed set $R$, $\T_d(R) = 0$.
\item[(b)] If $b < \ob$, then there is a
$\mu_d^*$-random closed set $R$ with $\T_d(R) > 0$.
\end{itemize}
\end{theorem}

\begin{proof}
Fix $d$ as described above so that $d(\sigma \fr i) = d(\sigma) \cdot
b$ and let $\mu^* = \mu_d^*$. We will compute the probability,
given two closed sets $Q$ and $K$, that $Q \cap K$ is nonempty.
Here we define the usual product measure
on the product space $\C \times \C$ of pairs $(Q,K)$ of nonempty closed sets
by letting $\mu^2(U_A \times U_B) = \mu^*(U_A) \cdot \mu^*(U_B)$
for arbitrary  subsets $A,B$ of $\{0,1\}^n$.

Let
\[
Q_n = \bigcup \{I(\sigma): \sigma \in \{0,1\}^n\ \&\ Q \cap I(\sigma)\neq \emptyset\}
\]
and similarly for $K_n$. Then $Q \cap K \neq \emptyset$ if and only if
$Q_n \cap K_n \neq \emptyset$ for all $n$. Let $p_n$ be the
probability that $Q_n \cap K_n \neq \emptyset$ for two arbitrary
closed sets $K$ and $Q$, relative to our measure $\mu^*$. It is
immediate that $p_1 = 1 - 2b^2$, since $Q_1 \cap K_1 = \emptyset$ only
when $Q_1 = I(i)$ and $K_1 = I(1-i)$. Next we will determine the
quadratic function $f$ such that $p_{n+1} = f(p_n)$. There are 9
possible cases for $Q_1$ and $K_1$, which break down into 4 distinct
cases in the computation of $p_{n+1}$.

\medskip

{\bf Case (i)}: As we have seen, $Q_1 \cap K_1 = \emptyset$ with
probability $1 - 2b^2$.

\medskip

{\bf Case (ii)}: There are two chances that $Q_1 = K_1 = I(i)$, each
with probability $b^2$ so that $Q_{n+1} \cap K_{n+1} \neq \emptyset$
with (relative) probability $p_n$.

\medskip

{\bf Case (iii)}: There are four chances where $Q_1 = \TN$ and $K_1
=I(i)$ or vice versa, each with probability $b \cdot (1-2b)$, so that
once again $Q_{n+1} \cap K_{n+1} \neq\emptyset$ with relative probability
$p_n$.

\medskip

{\bf Case (iv)}: There is one chance that $Q_1 = K_1 = \TN$, with
probability $(1 - 2b)^2$, in which case $Q_{n+1} \cap K_{n+1} \neq
\emptyset$ with relative probability $1 - (1 -p_n)^2 = 2p_n -
p_n^2$. This is because $Q_{n+1} \cap K_{n+1} = \emptyset$ if and only
if both $Q_{n+1} \cap I(i) \cap K_{n+1} = \emptyset$ for both $i=0$
and $i=1$.

\medskip

Adding these cases together, we see that
\[
p_{n+1} = [2b^2 + 4b (1-2b)] p_n + (1 - 2b)^2 (2p_n - p_n^2) = (2b^2 -
4b + 2) p_n  - (1 - 4b + 4b^2) p_n^2.
\]

Next we investigate the limit of the computable sequence $<p_n>_{n \in
 \omega}$. Let $f(p) =  (2b^2 - 4b + 2) p  - (1 - 4b + 4b^2) p^2$.
Note that $f(0) = 0$ and $f(1) = 1 - 4b^2 < 1$.
It is easy to see that the fixed points of $f$ are $p=0$ and $p =
 \frac{2b^2 -4b+1}{(1-2b)^2}$. Note that since $b < \frac 12$, the
 denominator is not zero and hence is always positive.

Now consider the function $g(b) = 2b^2 - 4b +1 = 2 (b-1)^2 - 1$, which has
positive root $\ob$ and is decreasing for $0 \leq b \leq 1$. There are
three cases to consider when comparing $b$ with $\ob$.

\medskip

{\bf Case 1}: If $b > \ob$, then $g(b) < 0$ and hence the other fixed
point of $f$ is negative. Furthermore, $2b^2 - 4b +2 < 1$ so that
$f(p) < p$ for all $p > 0$. It follows that the sequence $\{p_n: n \in
\N\}$ is decreasing with lower bound zero and hence must converge to a
fixed point of $f$ (since $p_{n+1} = f(p_n)$). Thus $lim_n p_n = 0$.

{\bf Case 2}: If $b = \ob$, then $g(b) = 0$ and $f(p) = p - (4b-1)
p^2$, so that $p=0$ is the unique fixed point of $f$. Furthermore,
$4b-1 = 3 - 2 \sqrt2 > 0$, so again $f(p) < p$ for all $p$. It follows
again that $lim_n p_n = 0$.

In these two cases, we can define a \ml test to prove that $T_d(R) =
0$
for any $\mu$-random closed set $R$.

For each $m, n \in \omega$, let
\[
B_m = \{(K,Q): K_m \cap Q_m \neq \emptyset\},
\]
so that $\mu^*(B_m) = p_m$ and let
\[
A_{m,n} = \{Q: \mu^*(\{K: K_m \cap Q_m \neq \emptyset\}) \geq
2^{-n}\}.
\]

\begin{claim} \label{c1} For each $m$ and $n$, $\mu^*(A_{m,n}) \leq 2^n \cdot p_m$.
\end{claim}

\emph{Proof of Claim \ref{c1}}. Define the Borel measurable function $F_m: \C
\times \C: \to \{0,1\}$ to be the characteristic function of
$B_m$. Then
\[
p_m = \mu^2(B_m) = \int_{Q \in \C} \int_{K \in \C} F(Q,K) dK dQ.
\]
Now for fixed $Q$,
\[
\mu^*(\{K: K_m \cap Q_m \neq \emptyset\}) = \int_{K \in \C} F(Q,K) dK,
\]
so that for $Q \in A_{m,n}$, we have $\int_{K \in \C} F(Q,K) dK \geq 2^{-n}$.
It follows that
\begin{align*}
p_m = \int_{Q \in \C} \int_{K \in \C} F(Q,K) dK dQ &\geq \int_{Q \in
  A_{m,n}} \int_{K \in \C} F(Q,K) dK dQ \\
&\geq \int_{Q \in A_{m,n}}
2^{-n} dQ = 2^{-n} \mu^*(A_{m,n}).
\end{align*}
Multiplying both sides by $2^n$ completes the proof of Claim \ref{c1}. $\qed$

\medskip

Since the computable sequence $<p_n>_{n \in \omega}$ converges to 0,
there must be a computable subsequence $m_0,m_1,\dots$ such that
$p_{m_n} < 2^{-2n-1}$ for all $n$. We can now define our \ml test. Let

\[
S_r = A_{m_r,r}
\]
and let
\[
V_n = \cup_{r>n} S_r.
\]
It follows that
\[
\mu^*(A_n) \leq 2^{n+1}\mu^*(B_{m_n}) < 2^{n+1} 2^{-2n-1} = 2^{-n}
\]
and therefore
\[
\mu^*(V_n) \leq \sum_{r>n} 2^{-r} = 2^{-n}
\]
Now suppose that $R$ is a random closed set. The sequence $\la V_n
\ra_{n \in \omega}$ is a computable sequence of c.e. open sets with
measure $\leq 2^{-n}$, so that there is some $n$ such that $R \notin
S_n$. Thus for all $r > n$, $\mu^*(\{K: K_{m_r} \cap R_{m_r} \neq
\emptyset\}) < 2^{-r}$ and it follows that
\[
\mu^*(\{K: K \cap R \neq \emptyset\}) = lim_n \mu^*(\{K: K_{m_n} \cap
R_{m_n} \neq \emptyset\}) = 0.
\]
Thus $\T_d(R) = 0$, as desired.

\medskip

{\bf Case 3}:
Finally, suppose that $b < \ob$. Then $0 < 2b^2 - 4b +1 < 1$, so that
$f$ has a positive fixed point $m_b = \frac{2b^2 -4b+1}{(1-2b)^2}$.
It is clear that $f(p) > p$ for $0 < p < m_b$ and $f(p) < p$ for $m_b
< p$. Furthermore, the function $f$ has its maximum at $p =
[\frac{1-b}{1-2b}]^2 > 1$, so that $f$ is monotone increasing on
$[0,1]$ and hence $f(p) > f(m_b) = m_b$ whenever $p > m_b$.  Observe that $p_0
= 1 > m_b$ and hence the sequence $\{p_n: n \in \N\}$ is decreasing
with lower bound $m_b$. It follows that $lim_n p_n = m_b > 0$.

Now $B = \{(Q,K): Q \cap K \neq \emptyset\} = \cap_n B_n$ is the
intersection of a decreasing sequence of sets and hence $\mu^2(B) =
lim_n p_ = m_b >0$.

\begin{claim} \label{c2} $\mu^*(\{Q: \mu^*(\{K: K \cap Q \neq
\emptyset\}) > 0\}) \geq m_b$.
\end{claim}

\emph{Proof of Claim \ref{c2}}. Let $B = \{(K,Q): K \cap Q \neq \emptyset$,
let $A = \{Q: \mu^*(\{K: K \cap Q \neq \emptyset\}) > 0\}$ and suppose
that $\mu^*(A) < m_b$.  As in the proof of Claim \ref{c1}, we have
\[
m_b = \mu^2(B) = \int_{Q \in \C} \int_{K \in \C} F(Q,K) dK dQ.
\]
For $Q \notin A$, we have $\int_{K \in Q} F(Q,K) dK = \mu^*(\{K: K
\cap Q \neq \emptyset\}) = 0$, so that
\[
m_b = \int_{Q \in A} \int_{K \in Q} F(Q,K) dK dQ \leq \int_{Q \in A} dQ = \mu^*(A),
\]
which completes the proof of Claim \ref{c2}. $\qed$

\begin{claim} \label{c3} $\{Q: \T_d(Q) \geq m_b\}$ has positive measure.
\end{claim}

\emph{Proof of Claim \ref{c3}}. Recall that $T_d(Q) = \mu^*(\{K: Q \cap K \neq \emptyset\})$.
Let $B = \{(K,Q): K \cap Q \neq \emptyset$, let $A = \{Q: T_d(Q) \geq m_b\}$
and suppose that $\mu^*(A) = 0$. As
in the proof of Claim \ref{c1}, we have
\[
m_b = \mu^2(B) = \int_{Q \in \C} T_d(Q) dQ.
\]
Since $\mu^*(A) = 0$, it follows that for any $B \subseteq \C$, we have
\[
\int_{Q \in B} T_d(Q) dQ \leq m_b \mu^*(B).
\]
Furthermore, $T_d(Q) < m_b$ for almost all $Q$, so there exists some $P$ with
$T_d(P) < m_b - \epsilon$ for some positive $\epsilon$. This means that for some
$n$, $\mu^*(\{K: P_n \cap K_n \neq \emptyset\}) < m_b - \epsilon$. Then for \emph{any}
closed set $Q$ with $Q_n = P_n$, we have $T_d(Q) < m_b - \epsilon$. But
$E = \{Q: Q_n = P_n\}$ has positive measure, say $\delta > 0$. Then we have

\begin{align*}
m_b = \int_{Q \in \C} T_d(Q) dQ &= \int_{Q \in E} T_d(Q) dQ\ +\ \int_{Q \notin E} T_d(Q) dQ \\
&\leq \ \delta (m_b - \epsilon) + (1- \delta) m_b = m_b - \epsilon \delta < m_b.
\end{align*}

This contradiction demonstrates Claim \ref{c3}. $\qed$

Since the set of $\mu*$-random closed sets has measure one, there must
be a random closed set $R$ such that $\T_d(R) \geq m_b$ and in
particular, there is a $\mu^*$-random closed set with positive capacity.
\end{proof}

Thus for certain measures, there exists a random closed set with
measure zero but with positive capacity. For the standard measure, a random closed set
has capacity zero.

\begin{corollary}
\label{th3}
Let $d$ be the uniform measure with $b_0 = b_1 = b_2 = \frac13$. Then
for any $\mu_d^*$-random closed set $R$, $\T_d(R) = 0$.
\end{corollary}

A random closed set may not be effectively closed. But we can also
construct an effectively closed set with measure zero and with
positive capacity.

\begin{theorem}  \label{th5}
For the regular measure $\mu_d$ with $b = b_1 = b_2$, there is a
$\Pi^0_1$ class $Q$ with Lebesgue measure zero and positive capacity
$\T_d(Q).$
\end{theorem}

\begin{proof}
First let us compute the capacity of $X_n = \{x: x(n) =0\}$. For
$n=0$, we have $\T_d(X_0) = 1 - b$. That is, $Q$ meets $X_0$ if and
only if $Q_0 = I(0)$ (which occurs with probability $b$), or $Q_0 =
\TN$ (which occurs with probability $1 - 2b$. Now the probability
$\T_d(X_{n+1})$ that an arbitrary closed set $K$ meets $X_{n+1}$ may
be calculated in two distinct cases.  As in the proof of Theorem
\ref{th3}, let
\[
K_n = \bigcup \{I(\sigma): \sigma \in \{0,1\}^n\ \&\ K \cap I(\sigma)\neq \emptyset\}
\]

{\bf Case I} If $K_0 = \TN$, then $\T_d(X_{n+1}) = 1 - (1- \T_d(X_n))^2$.

\medskip

{\bf Case II} If $K_0 = I((i))$ for some $i<2$, then $\T_d(X_{n+1}) = \T_d(X_n)$.

\medskip

It follows that
\begin{align*}
\T_d(X_{n+1}) &= 2b \T_d(X_n) + (1-2b) (2 \T_d(X_n)
- (\T_d(X_n))^2) \\
&= (2-2b) \T_d(X_n) - (1-2b) (\T_d(X_n))^2
\end{align*}

Now consider the function $f(p) = (2-2b) p - (1-2b) p^2$, where $0 < b
 < \frac 12$. This function has the properties that $f(0) = 0$, $f(1) =
 1$ and $f(p) > p$ for $0 < p < 1$. Since $\T_d(X_{n+1})
= f(\T_d(X_n))$, it follows that $lim_n \T_d(X_n) = 1$ and is the limit
of a computable sequence.

For any $\sigma = (n_0, n_1, \dots, n_k)$, with $n_0 < n_1 < \cdots
<n_k$, similarly define $X_{\sigma} = \{x: (\forall i<k) x(n_i) =
0\}$.  A similar argument to that above shows that $lim_n
\T_d(X_{\sigma \fr n}) / \T_d(X_{\sigma}) = 1$.

Now consider the decreasing sequence $c_k = \frac{2^{k+1} +
1}{2^{k+2}}$ with limit $\frac12$.  Choose $n = n_0$ such that
$\T_d(X_n) \geq \frac 34 = c_0$ and for each $k$, choose $n = n_{k+1}$
such that $\T_d(X_{(n_0,\dots,n_k,n)}) \geq c_{k+1}$. This can be done
since $c_{k+1} < c_k$. Finally, let $Q = \bigcap_k
X_{(n_0,\dots,n_k)}$.  Then $\T_d(Q) = lim_k \T_d(X_{(n_0,\dots,n_k)})
\geq lim_k c_k=\frac12$.
\end{proof}

\section{Conclusions and Future Research}

In this paper, we have established a connection between measure and
capacity for the space $\C$ of closed subsets of $\TN$. We showed that
for a computable
measure $\mu^*$, a computable capacity may be defined by letting
$\T(Q)$ be the measure of the family of closed sets $K$ which have
nonempty intersection with $Q$. We have proved an effective
version of the Choquet's theorem by showing that every computable
capacity may be obtained from a computable measure in this way.

For the uniform measure $\mu$ under which a node $\sigma$ in $T$ has
exactly one immediate extension $\sigma \fr i$ with probability $b$
for $i=0,1$ (and hence $\sigma$ has both extensions with probability
$1-2b$), we have established conditions on $b$ that characterize when
the capacity of a random closed set equals zero or is $>0$. We have
also constructed for each such measure an effectively closed set with
positive capacity and with Lebesgue measure zero.

In future work, we plan to extend our results to more general measures
where for each string $\sigma \in T_Q$, the probability that $\sigma
\fr i \in T_Q$ depends on $\sigma$. For example, such a measure on the
space of closed sets may be defined  by making the probability that
both extensions $\sigma \fr i$ of a node $\sigma \in T$ belong to $T$
equal to $1 - \frac 2n$ and the probability that just one extension
belongs to $T$ equal to $\frac 1n$, where $n = |\sigma|$.


\begin{thebibliography}{99}

\bibitem{Ax10} L. Axon, \emph{Algorithmically Random Closed Sets and Probability},
Ph. D. Dissertation, Notre Dame University (2010).

\bibitem{BCRW08} G. Barmpalias, P. Brodhead, D. Cenzer, J.B. Remmel,
R. Weber, \emph{Algorithmic Randomness of Continuous Functions},
Archive for Mathematical Logic 46 (2008), 533--546.

\bibitem{BCD06} P. Brodhead, D. Cenzer and S. Dashti, \emph{Random
closed sets}, in \emph{Logical Approaches to Computational Barriers},
eds. A. Beckmann, U. Berger, B. L\"{o}we and J.V. Tucker, Springer
Lecture Notes in Computer Science 3988 (2006), 55--64.

\bibitem{BCDW08} G. Barmpalias, P. Brodhead, D. Cenzer, S. Dashti and
R. Weber, \emph{Algorithmic randomness of closed sets}, J. Logic and
Computation 17 (2007), 1041--1062.

\bibitem{BCR06} P. Brodhead, D. Cenzer and J.~B. Remmel, \emph{Random
continuous functions}, in \emph{CCA 2006, Third International
Conference on Computability and Complexity in Analysis},
eds. D. Cenzer, R. Dillhage, T. Grubb and Klaus Weihrauch, Information
Berichte, FernUniversit\"{a}t (2006), 79--89 and Springer Electronic
Notes in Computer Science (2006).

\bibitem{CR99} D. Cenzer and J.~B. Remmel, \emph{$\Pi^0_1$ classes},
in \emph{Handbook of Recursive Mathematics, Vol. 2: Recursive Algebra,
Analysis and Combinatorics}, editors Y. Ersov, S. Goncharov, V. Marek,
A. Nerode, J. Remmel, Elsevier Studies in Logic and the Foundations of
Mathematics, Vol. 139 (1998) 623--821.

\bibitem{Ceta} D. Cenzer, \emph{$\Pi^0_1$ Classes}, ASL Lecture Notes
  in Logic, to appear.

%\bibitem{CR99} D. Cenzer and J.~B. Remmel, \emph{$\Pi^0_1$ classes},
%in \emph{Handbook of Recursive Mathematics, Vol. 2: Recursive Algebra,
%Analysis and Combinatorics}, editors Y. Ersov, S. Goncharov, V. Marek,
%A. Nerode, J. Remmel, Elsevier Studies in Logic and the Foundations of
%Mathematics, Vol. 139 (1998) 623--821.

%\bibitem{Ch75}[Ch75]
%G. Chaitin, \emph{A theory of program size formally identical to
%information theory}, J. ACM 22 (1975), 329-340.

%\bibitem{Ch76}[Ch76]
%G. Chaitin, \emph{Information-theoretical characterizations of
%recursive infinite strings}, Theor. Comp. Sci. 2 (1976), 45-48.

\bibitem{D77} C. Dellacheries, \emph{Les d\'{e}rivations en th\'{e}orie
  descriptive des ensembles et le th\'{e}or\`{e}me de la borne}. S\'{e}minaire de
  Probabilities XI, Universit\'{e} de Strasbourg, Springer Lecture Notes
  in Mathematics, vol. 581 (1977), 34-46.

%\bibitem{DGL04}[DGL04]
%R. Downey, E. Griffths, G. Laforte, \emph{On Schnorr and computable
%randomness, martingales, and machines}, Math. Logic. Quart. 50,
%(2004), 613-627.

\bibitem{DH:book} R.\ Downey and D.\ Hirschfeldt, \emph{Algorithmic
Randomness and Complexity}, Springer-Verlag, in press.

%\bibitem{DHL01} R.~Downey; D.~Hirschfeldt; G.~LaForte, Randomness and
%reducibility.  Mathematical foundations of computer science, 2001,
%316--327, {\em Lecture Notes in Comput.\ Sci.}, 2136, Springer,
%Berlin, 2001.

\bibitem{GMW88} S. Graf, R.D. Mauldin and S.C. Williams,
\emph{The exact Hausdorff dimension in random recursive constructions},
Memoirs Amer. Math. Soc. 381 (1988).

\bibitem{KH09} B. Kjos-Hanssen,
\emph{Infinite subsets of random sets of integers}, Math. Research Letters
16 (2009), 103-110.

\bibitem{DK09} D. Diamondstone and B. Kjos-Hanssen, \emph{Members of random
closed sets}, in CIE 2009 (eds. K. Ambos-Spies, B. Loewe and W. Merkle),
Lecture Notes in Computer Science 5635 (2009), 144-153.

%\bibitem{K65}
%A.~N. Kolmogorov, \emph{Three approaches to the quantitative
%definition of information}, in \emph{Problems of Information
%Transmission, Vol. 1} (1965), 1-7.

%\bibitem{Kucera:1985} A. Ku\v{c}era, \emph{Measure, $\Pi^0_1$ classes,
%and complete extensions of PA}, \emph{Springer Lecture Notes in
%Mathematics} Vol.\ 1141 (1985), 245-259.

%\bibitem{L73} L. Levin, \emph{On the notion of a random sequence},
%Soviet Mat. Dokl. 14 (1973), 1413-1416.

%\bibitem{ML66} P. Martin-L\"{o}f,
%\emph{The definition of random sequences}, Information and Control 9
%(1966), 602-619.

\bibitem{MM09} A. McLinden and R.D. Mauldin, \emph{Random closed sets
  viewed as random recursions}, Archive for Math. Logic 48 (2009),
  257-263.

\bibitem{Ng06} H.~T.~Nguyen, \emph{An Introduction to Random Sets},
Chapman and Hall (2006).

\bibitem{Nies:book} A.\ Nies, \emph{Computability and Randomness},
Oxford University Press (2009).

\bibitem{RSta} J. Reimann and T. Slaman, \emph{Measures and their
  random reals}, Transactions Amer. Math. Soc., to appear.

%\bibitem{Schnorr}
%C.-P. Schnorr, \emph{Zuf\"{a}lligkeit und Wahrscheinlichkeit}, Lecture
%Notes in Mathematics 218, Springer-Verlag, Berlin-Heidelberg-New
%York-Tokoyo (1971).

%\bibitem{vM19} R. von Mises, \emph{Grundlagen der
%Wahrscheinlichkeitsrechnung}, %Math. Zeitschrift 5 (1919), 52-99.

%\bibitem{Zambella:1990} D. Zambella, \emph{On sequences with simple
%initial segments}, ILLC technical report ML-1990-05, University of
%Amsterdam (1990).

\end{thebibliography}
\end{document}